\documentclass[notitlepage,%
 reprint,
 amsmath,amssymb,
 aps,
]{revtex4-1}

\usepackage[utf8]{inputenc}
\usepackage{amsmath,amssymb,amsthm}
\usepackage{mathtools}

\usepackage{qcircuit}
\usepackage{color}

\usepackage{algorithm}
\usepackage[noend]{algpseudocode}

\makeatletter
\newcommand{\algorithmlabel}[2]{%
\protected@write \@auxout {}{\string \newlabel {#1}{{#2}{}}}}
\makeatother

\newsavebox{\algbox}

\newcommand{\sfrac}[2]{{\ensuremath{\textstyle\frac{#1}{#2}}}}
\newcommand{\half}[0]{\sfrac{1}{2}}

\newcommand{\ket}[1]{\ensuremath{\vert{#1}\rangle}}
\newcommand{\bra}[1]{\ensuremath{\langle{#1}\vert}}
\newcommand{\braket}[2]{\ensuremath{\langle{#1}\vert{}{#2}\rangle}}
\newcommand{\kett}[1]{\ensuremath{\vert{#1}\rangle\!\rangle}}

\newcommand{\abs}[1]{\ensuremath{\vert{#1}\vert}}
\newcommand{\diag}[0]{\ensuremath{\mbox{diag}}}
\newcommand{\comm}[2]{\ensuremath{\gamma_{#1,#2}}}

\newcommand{\Tr}[0]{\ensuremath{\mathrm{Tr}}}
\renewcommand{\Pr}[0]{\ensuremath{\mathrm{Pr}}}

\newtheorem{theorem}{Theorem}[section]
\newtheorem{lemma}[theorem]{Lemma}

\definecolor{gray}{rgb}{0.4,0.4,0.4}

\usepackage{graphicx}
\usepackage{dcolumn}
\usepackage{bm}
\usepackage[hidelinks]{hyperref}

\newcommand{\PauliX}{\mbox{Pauli-X}}
\newcommand{\PauliZ}{\mbox{Pauli-Z}}

\newcommand{\AMatrix}[2]{\ensuremath{A_{{#1},{#2}}}}
\newcommand{\AStarMatrix}[2]{\ensuremath{A_{{#1},{#2}}^{\star}}}

\begin{document}

\title{Model-free readout-error mitigation for quantum expectation values}

\author{Ewout van den Berg}%
\author{Zlatko~K.~Minev}%
\author{Kristan Temme}%
\affiliation{%
 IBM Quantum, T.J.~Watson Research Center\\
 Yorktown Heights, NY, USA
}%

\date{\today}
\begin{abstract}
  Measurements on current quantum processors are subject to hardware
  imperfections that lead to readout errors. These errors manifest
  themselves as a bias in quantum expectation values. Here, we propose
  a very simple method that forces the bias in the expectation value
  to appear as a multiplicative factor that can be measured directly
  and removed at the cost of an increase in the sampling complexity
  for the observable. The method assumes no specific form of the
  noise, but only requires that the noise is `weak' to avoid excessive
  sampling overhead. We provide bounds relating the error in the
  expectation value to the sample complexity.
\end{abstract}

\maketitle

\section{Introduction}

Quantum
algorithms~\cite{peruzzo2014variational,o2016scalable,KAN2017MTTa,larose2019variational,havlivcek2019supervised,schuld2019quantum,mcardle2019variational,mitarai2020theory}
for near-term devices can often be described by the execution of a
reasonably-shallow quantum circuit followed by the measurement of an
observable through sampling. A general assumption for near-term
devices is that proper quantum
error-correction~\cite{SHO1995a,gottesman1997stabilizer,DEV2013MNa,lidar2013quantum}
is not yet available. As a result, noise parameters, such as coherence
time, dictate the maximum depth of a circuit and therefore determine
the size of the calculation that can be performed. Even when working
within these limitations, hardware noise can still affect expectation
values in the form of a bias. Error-mitigation techniques have
therefore been introduced to remove this bias and produce more
accurate expectation values. These techniques come at the additional
cost of repeating the computation, possibly with altered parameters,
increased sampling cost, or additional classical post processing.  For
the mitigation of errors that occur during the application of the
quantum circuit, several schemes have been
proposed~\cite{temme2017error,li2017efficient,bonet2018low,endo2018practical,mcclean2020decoding,arXiv:2011.01382,lowe2020unified,koczor2020exponential,huggins2020virtual}
and implemented experimentally~\cite{Kandala2019,Songeaaw5686}.

In this work, we consider the mitigation of readout errors that occur
during the final measurement step of the computation. We focus on the
computation of the expectation values of Pauli observables. Since
Pauli matrices constitute a Hermitian matrix basis, they can represent
any observable~\cite{PAR2004Ra}.  Moreover, any observable that can be
expanded in a polynomial number of Pauli matrices, such as local
Hamiltonians, can be estimated efficiently by measuring the
expectation values of Pauli observables, due to linearity of the
expectation value.

After a quantum circuit has been applied, we can measure a Pauli
observable. This is done by rotation of the observable to the
computational basis using single-qubit Clifford gates, followed by
measurement in this basis and some basic classical post processing. In
the absence of readout errors, the measurement output of $n$ qubits is
fully described by a probability distribution $p$ over the $2^n$
computational basis states. The standard model for readout errors is
given by a classical noise
map~$A$~\cite{GEL2020a,MAC2019ZOa,Haapasalo2012}, which maps the
noise-free $p$ to the noisy $\tilde{p}$ readout distribution by
$\tilde{p} = Ap$. The readout map $A$ is a $2^n$-by-$2^n$
left-stochastic matrix, where the entry $A_{i,j}$ denotes the
probability of measuring the $i$-th instead of the $j$-th
computational basis state, for $i,j \in \{0,1\}^{n}$.

A direct approach to mitigate the effect of read-out errors has
frequently been to estimate columns of $A$ by measuring the output
frequencies $\hat{r}_x$ for different bit strings $x$ and then to
apply the matrix $A^{-1}$ for
mitigation~\cite{MAC2019ZOa,Steffen1423,KAN2017MTTa}. Explicit
representation and inversion of $A$ is of course feasible only when
the system size is small, or when the noise can be assumed to
factorize such that noise on individual or small groups of qubits can
be modeled and inverted independently. However, experiments have shown
that the noise tends to be correlated~\cite{CHE2019FYWa}, which
invalidates the use of product approximations to the stochastic
matrix. Several approaches have been proposed to deal with this more
difficult scenario~\cite{NAC2019UJBa-arXiv,9259938,KWO2020Ba} as well
as with other
settings~\cite{10.1145/3352460.3358265,HIC2020BNa,KWO2020Ba}. Recently,
a readout error-mitigation scheme for correlated noise with a formal
performance guarantee on sampling overhead depending on the noise
strength was proposed and implemented experimentally
in~\cite{BRA2021SKMa}. In this approach, the noise map does not
need to be explicitly inverted and the model can be concisely
represented using ${\cal O}(\mbox{poly}(n))$ parameters.

Here, we propose a readout-error mitigation method that is motivated
by work on quantum benchmarking
protocols~\cite{ERH2019WPMa,FLA2019Wa-arXiv}.  The method, introduced
in Section~\ref{Sec:Approach}, randomizes the output channel by
uniformly applying random Pauli bit flips prior to measurement, which
are tracked and used in the subsequent analysis. In
Section~\ref{Sec:Derviation} we show that this randomization
transforms the action of an arbitrary noise map $A$ into a single
multiplicative factor per Pauli observable; that is, it diagonalizes
the measurement channel. The multiplicative factors can be measured
directly, in the absence of the quantum circuit. By dividing out this
factor, the bias-free mitigated Pauli expectation value is
obtained. The method does not require a model of the physical
measurement noise and does not make assumptions on the noise
strength. In fact, the method only makes the assumption that the
circuit can be initialized in the zero state.

The shot-noise variance of the mitigated estimate scales inversely
proportional with the magnitude of noise factor. In
Section~\ref{Sec:Analysis}, we show how this magnitude depends on the
underlying noise strength and analyze the required number of
measurement samples to attain a desired estimation
accuracy. Simulations in Section~\ref{Sec:Simulation} indicate that
the method scales reasonably to larger system sizes in the presence of
moderate noise.

\section{Method}\label{Sec:Approach}
Consider a system of $n$ qubits and order the set of Pauli operators
such that $P_q$ denotes a unique Pauli operator indexed by
$q \in \mathcal{P} := [0,4^{n-1}]$. The \PauliZ\ operators are
assigned indices $\mathcal{Z} := [0,2^{n-1})$ and the set of indices
corresponding to \PauliX\ operators is denoted $\mathcal{X}$.  Given
$r,s\in \mathbb{Z}_2^n$ with inner product
$\langle r, s\rangle = \sum_i r_is_i$, we define
$Z^s = \bigotimes_{i=1}^n \sigma_z^{s_i}$ and likewise for $X_s$ with
$\sigma_z$ replaced by $\sigma_x$. Interpreting $s$ as an
integer index, we set $P_s = Z^s$.

Starting from the initial state
\begin{equation}\label{Eq:Rho0}
\rho_0 = \ket{0}\bra{0} = 2^{-n}(I + \sigma_z)^{\otimes n} =
2^{-n}\sum_{j\in\mathcal{Z}} P_{j},
\end{equation}
we want to estimate the \PauliZ\ component $P_i$ in state
$\rho = U\rho_0 U^{\dag}$ obtained by applying operator $U$, namely
\begin{equation}\label{Eq:QST}
\langle P_i\rangle_\rho = \Tr(P_i\rho).
\end{equation}
We assume that the initial state is $\rho_0$ and that all measurements
are performed in the computational basis. This means that we can only
evaluate~\eqref{Eq:QST} for $i\in \PauliZ$. Expectation values for
other Paulis can be obtained by incorporating an appropriate basis
change in operator $U$.

\begin{figure}
\centering
\begin{tabular}{ccc}
\includegraphics[height=16pt]{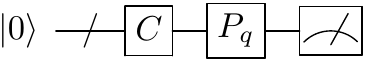}
& \hspace*{22pt} &
\includegraphics[height=16pt]{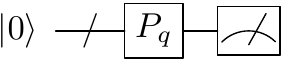}\\[2pt]
({\bf{a}}) & & ({\bf{b}})
\end{tabular}
\caption{Illustration of the (a) measurement and (b) calibration
  circuits  used for the estimation of error-mitigated
  averages. }\label{Fig:Circuit}
\end{figure}

In order to estimate~\eqref{Eq:QST} we run various instances of the
circuit given in Figure~\ref{Fig:Circuit}(a). The circuit is
parameterized by Pauli index $q$, and unitary $C$ (or the circuit
that implements it). Choosing $C$ to be the identity results in a
simplified circuit, as shown in Figure~\ref{Fig:Circuit}(b). The Pauli
index $q$ will be sampled from some index set $\mathcal{S}$ to be
specified later. The protocol to acquiring measurement
outcomes for $N$ circuit instances is given by:

\medskip
\begin{algorithmic}[1]
    \Statex \hspace*{-12pt}{{\bf{Protocol}} {\bf{AcquireData}}($\mathcal{S}$, $C$, $N$)}
    \State Initialize an empty data set $\mathcal{D}$
    \For{$i = 1,\ldots,N$}
       \State Uniformly sample $q$ from index set $\mathcal{S}$
       \State Execute the circuit in Figure~\ref{Fig:Circuit}(a) with Pauli
       \Statex \hspace*{15pt}$P_q$ and unitary $C$\vspace*{1pt}
       \State Record the measurement outcome $x$ and add
       \Statex \hspace*{15pt}tuple $(q,x)$ to $\mathcal{D}$
    \EndFor
    \State \textbf{return} $\mathcal{D}$
\end{algorithmic}

\medskip\noindent Each measurement outcome is represented by an
element $x \in \mathbb{Z}_2^n$.
In classical post-processing of the acquired data, we use the function
\begin{equation}\label{Eq:Processing}
f(\mathcal{D},s) =
\frac{1}{\vert\mathcal{D}\vert}\sum_{(q,x)\in\mathcal{D}}
\comm{s}{q}(-1)^{\langle s, x\rangle},
\end{equation}
where $\comm{a}{b}$ has the value $1$ if Paulis $P_a$ and $P_b$
commute and the value $-1$ otherwise (these sign changes with respect
can be omitted if we flip measurement bits according to the sampled
$q$ value).  The protocol for estimating $\langle Z^s\rho)$ is then as follows:\\

{\bf{Protocol 1}}\\[-18pt]
\begin{enumerate}\setlength{\itemsep}{-0.25pt}
\item $\mathcal{D}_0 = \mbox{{AcquireData}}(\mathcal{X},I,N)$
\item $\mathcal{D}_1 = \mbox{{AcquireData}}(\mathcal{X},U,N)$
\item Return estimate $f(\mathcal{D}_1,s) / f(\mathcal{D}_0,s)$
\end{enumerate}
\smallskip

\noindent Note that the data acquired in steps 1 and 2 can be reused
to evaluate the quantity in step 3 for different values of
$s$. Moreover, the data from step 1 is independent of $U$ and can
therefore be used in error mitigation of measurement of other states
as well. For simplicity we set the number of samples in each of the
two data sets to $N$. More generally we could choose different numbers
of samples for each of these steps.

\section{Derivation}\label{Sec:Derviation}

Ideal measurements in the computational basis can be written in terms
of positive operator-valued measures $E_x = \ket{x}\bra{x}$ for
$x\in\mathbb{Z}_2^n$. We assume that
measurements are affected by a noise map $A$, such that measurement
$y$ can be misinterpreted as $x$ with probability $\AMatrix{x}{y} =
\bra{x}A\ket{y}$. Using this, we can define noisy measures $\tilde{E}_x =
\sum_y \AMatrix{x}{y}\ket{y}\bra{y}$. Now, for $s\in\mathbb{Z}_2^n$, define
\begin{align*}
X_s &:= \sum_a \ket{a+s}\bra{a} = \sum_a \ket{a}\bra{a+s} =
X_s^{\dag},\ \mbox{and}\\
Z_s &:= \sum_a (-1)^{\langle s, a\rangle}\ket{a}\bra{a}.
\end{align*}
We would like to estimate the expectation value
\begin{equation}\label{Eq:IdealZs}
\langle Z_w\rangle_{\rho} = \Tr(Z_w\rho) = \sum_{x\in\mathbb{Z}_2^n} (-1)^{\langle w,x\rangle}\Tr(E_x\rho).
\end{equation}
Substituting $E_x$ by $\tilde{E}_x$ gives an unmitigated noisy
estimate $\langle\tilde{Z}_s\rangle_{\rho}$. In order to mitigate the
readout error, our algorithm applies a random bit flips prior to
measurement, and then either applies the same bit flips directly after
the (noisy) measurement, or equivalently adjust signs in the
estimation of the expectation value. The random bit flip, obtained by
applying $X_s$ for a randomly sampled $s$, can be applied before the
noisy measurement, or directly after an ideal measurement but just
prior to the noise map $A$. Using the latter view, we can
define the twirled noise map $A^{\star}$ as
\begin{align*}
A^{\star} &:= \frac{1}{2^n}\sum_s X_s AX_s^{\dag}
= \frac{1}{2^n}\sum_s\sum_{a,b}
 \AMatrix{a}{b} X_s \ket{a}\bra{b}X_s^{\dag}\\
& =\frac{1}{2^n}\sum_s\sum_{a,b} \AMatrix{a}{b} \ket{a+s}\bra{b+s},
\end{align*}
with associated measure
$\tilde{E}_x^{\star} =
\sum_y \AStarMatrix{x}{y}\ket{y}\bra{y}$. Substitution
in~\eqref{Eq:IdealZs} then gives us the twirled noisy expectation
\begin{align}
\langle \tilde{Z}_w^{\star}\rangle_{\rho}
&:= \sum_{x\in\mathbb{Z}_2^n}
 (-1)^{\langle w,x\rangle}\Tr(\tilde{E}_x^{\star}\rho)\notag\\
&
= \sum_{x,y} (-1)^{\langle w,x\rangle}
\bra{x}A^{\star}\ket{y}\Tr(\ket{y}\bra{y}\rho)
\label{Eq:TwirledZs}
\end{align}
In order to simplify this, first define
\[
\ket{v_w} = \sum_x (-1)^{\langle w,x\rangle}\ket{x}.
\]
We then have
\begin{align*}
A^{\star}\ket{v_w}
& =\frac{1}{2^n}\sum_{s,x}\sum_{a,b} (-1)^{\langle w,x\rangle}\AMatrix{a}{b}
\ket{a+s}\braket{b+s}{x}
\\
& \hspace*{-11pt}\stackrel{(s=b+x)}{=}\frac{1}{2^n}\sum_{x}\sum_{a,b}
     (-1)^{\langle w,x\rangle}\AMatrix{a}{b}
\ket{a+b+x}\\
& = \frac{1}{2^n}\sum_{x}\sum_{a,b} (-1)^{\langle w,x+a+b\rangle}\AMatrix{a}{b}
\ket{x}\\
&= \lambda_w \ket{v_w},
\ \mbox{with}\ \lambda_w = \frac{1}{2^n}\sum_{a,b} (-1)^{w,a+b}\AMatrix{a}{b}.
\end{align*}
In other words, $\ket{v_w}$ is an (unnormalized) eigenvector of
$A^{\star}$ with corresponding eigenvalue $\lambda_w$. We therefore
have
$\bra{v_w}A^{\star} = \lambda_w\bra{v_w}$, and it immediately follows
that we can rewrite~\eqref{Eq:TwirledZs} as
\begin{align*}
\langle \tilde{Z}_w^{\star}\rangle_{\rho}
& = \bra{v_w}A^{\star}\sum_y \ket{y}\Tr(\ket{y}\bra{y}\rho)
\\
& = \lambda_w\bra{v_w}\sum_y \ket{y}\Tr(\ket{y}\bra{y}\rho)
\\
& = \lambda_w\sum_{x,y}(-1)^{\langle w,x\rangle}\braket{x}{y}\Tr(\ket{y}\bra{y}\rho)
\\
& = \lambda_w\sum_{x}(-1)^{\langle w,x\rangle}\Tr(\ket{x}\bra{x}\rho)
\\
& = \lambda_w \langle Z_w\rangle_{\rho}.
\end{align*}
For the initial state $\rho = \ket{0}\bra{0}$ we have
$\langle Z_w\rangle_{\rho} = 1$ and therefore
$\langle \tilde{Z}_w^{\star}\rangle_{\rho} = \lambda_w$.  The protocol
estimates this quantity, and then uses it to obtain noise-mitigated
estimates $\langle Z_w\rangle_{\rho}$ for other values of $\rho$.

\subsection{Alternative derivation}

As an alternative derivation, consider the super-operator
representation of the state in the Pauli basis as $\kett{\rho}$. In
the case of ideal measurements in the computational basis, the readout
probabilities are given by projection operator $[H_n^{-1}, 0]$ with
$H_n = H^{\otimes n}$, where $H$ the unnormalized Hadamard matrix
$X+Z$. The noisy readout probabilities are then given by the vector
\[
A \left[H_n^{-1}, 0\right]\kett{\rho}.
\]
Conversion of binary measurements to Pauli-Z observables is done
through the Walsh-Hadamard transformation. By appropriately ordering
the Pauli-Z operators, this amounts to multiplication by $H_n$. The
vector of Pauli-Z observable expectation values is then given by
\[
H_nA \left[H_n^{-1}, 0\right]\kett{\rho} = \left[ H_nAH_n^{-1}, 0\right]\kett{\rho}
= \left[ M, 0\right]\kett{\rho}
\]
with $M = H_nAH_n^{-1}$.
For the proposed error-mitigation scheme we add a random Pauli-X
operator $P_q$ prior to measurement and appropriate sign changes to the
estimate. The effect of this is multiplication with diagonal matrices
$D_q$ and $D_q'$ with diagonal elements $\comm{v}{q} = (-1)^{\langle v,q\rangle}$ for $v$ in
$\mathcal{Z}$ and $\mathcal{P}\setminus\mathcal{Z}$, respectively:
\[
D_q\left[ M, 0\right]\left[\begin{array}{cc}D_q & 0\\0 & D_q'\end{array}\right]\kett{\rho}
= \left[ D_qM D_q, 0\right]\kett{\rho}
\]
Multiplication from the left and right by a diagonal matrix
$\mathrm{diag}(d_q)$ amounts to elementwise multiplication with matrix
$d_qd_q^T$. The expectation of these matrices over the Pauli-X group
satisfies $\mathbb{E}_{q\in\mathcal{X}} [d_qd_q^T] = I$. The
expected observable vector is therefore given by 
\[
\left[M\odot I, 0\right]\kett{\rho}.
\]
We can determine $\lambda = \mathrm{diag}(M)$ by multiplying with
$\kett{\rho_0}$, which is a vector whose first $2^n$ elements are
one and all remaining elements are zero. Once we know $\lambda$ we can
easily divide it out to obtain unbiased Pauli-Z estimates.



\section{Analysis}\label{Sec:Analysis}

\subsection{Sample complexity}

Protocol 1 estimates
$\langle \tilde{Z}^s\rangle_{\rho}$ as $f(\mathcal{D}_1,s)/f(\mathcal{D}_0,s)$, which is of the form
$\hat{x}/\hat{y}$. We now consider the sample complexity of the
protocol: what value of $N$ we should choose, such that with
probability at least $1-\delta$ the final estimate deviates at most
$\epsilon$ from the exact value? Before doing so, we first consider
the accuracy of the estimate in the case we can estimate $x$ and $y$
up to an additive error of at most $\alpha$.
\begin{lemma}\label{Lemma:FractionXY}
  Let $x,y$ be such that $0 \leq \abs{x} \leq \abs{y} \leq 1$. Given
  estimates $\hat{x},\hat{y}$ with $\abs{x-\hat{x}} \leq \alpha$ and
  $\abs{y-\hat{y}} \leq \alpha$, such that
  $0 \leq \alpha \leq \abs{y}/2$. Then
\[
\left\vert\frac{\hat{x}}{\hat{y}}-\frac{x}{y}\right| \leq \frac{4\alpha}{y}.
\]
\end{lemma}
\begin{proof}
  Assume without loss of generality that $x,y\geq 0$. Taking the
  Taylor-series expansion around zero for sufficiently small $\alpha$
  we have in the worst case that
\begin{align*}
\frac{\hat{x}}{\hat{y}} &= \frac{x+\alpha}{y-\alpha}
=\frac{x}{y} + \left(1+\frac{x}{y}\right)\sum_{k=1}^{\infty}\left(\frac{\alpha}{y}\right)^k\\
&\leq
\frac{x}{y} + \left(1+\frac{x}{y}\right)\cdot\left(\frac{1}{1-\alpha/y}-1\right)\\
&\leq 
\frac{x}{y} + \left(1+\frac{x}{y}\right)\cdot\left(\frac{\alpha/y}{1-\alpha/y}\right)
\leq \frac{x}{y} + \frac{4\alpha}{y}.
\end{align*}
In the last inequality we use the fact that $x/y\leq 1$, and
$1-\alpha/y \geq 1/2$. A lower bound can be derived similarly to
obtain the given result.
\end{proof}

\noindent
With this we can obtain the following sample complexity:

\begin{theorem}
  With probability at least $1-\delta$, protocol 1
  estimates~\eqref{Eq:QST} with error at most $\epsilon$ for a fixed
  $i\in\mathcal{Z}$ when the number of samples $N$ satisfies
\[
N \geq \frac{32\log(4/\delta)}{M_{i,i}^2\epsilon^2}.
\]
\end{theorem}
\begin{proof}
  Protocol 1 acquires data and estimate different quantities
  using the function in~\eqref{Eq:Processing}. For a fixed $i$ and
  $j$, we can view each term in the summation as an independent
  $\pm 1$ sample from a certain distribution depending on $U$ that
  marginalizes over Pauli indices $p$ and $q$. For the error in the
  estimated quantities, we therefore apply Hoeff\-ding's inequality,
  which states that, given independent random variables $X_i$ from any
  distribution over $[-1,1]$, the deviation of
  $\bar{X} = N^{-1}\sum_{i=1}^N X_i$ to the expected value
  $\mathbb{E}(X)$ satisfies
\begin{equation}\label{Eq:HoeffdingA}
\Pr\left(\Big\vert \bar{X} - \mathbb{E}(X)\Big\vert \geq
  \alpha\right) \leq 2\exp\left(-\half N\alpha^2\right).
\end{equation}
We want to ensure that probability of deviating from the expectation
by $\alpha$ or more, is bounded by $\delta/2$. Using the union bound
it then follows that the enumerator and denominator are $\alpha$ close
to their expectation with probability at least $1-\delta$. Bounding
the failure probability in~\eqref{Eq:HoeffdingA} from above by
$\delta/2$ gives the sufficient condition
\begin{equation}\label{Eq:IntermediateNValue}
N \geq \frac{2\log(4/\delta)}{\alpha^2}.
\end{equation}
We now need to choose $\alpha$ such that the final estimate is
$\epsilon$ accurate. From Lemma~\ref{Lemma:FractionXY} we see that it
suffices to take $4\alpha / y \leq \epsilon$, where
$y = f_I(i) = M_{i,i}$. Substituting $\alpha = \epsilon M_{i,i}/4$
in~\eqref{Eq:IntermediateNValue} then gives the desired result.
\end{proof}

As discussed in more detail in Section~\ref{Sec:ExampleMatrices}, the
term $M_{i,i}$ is expected to scale weakly exponential in the weight
of the \PauliZ\ observable with a base that deviates from unity by the
magnitude of the noise. The increase in sampling complexity therefore
depends on the strength of the noise similar to the
quasi-probabilistic noise cancellation method in
ref.~\cite{temme2017error}.

\subsection{Number of circuit instances}

\begin{figure}
\centering
\begin{tabular}{cccc}
\includegraphics[width=0.21\columnwidth]{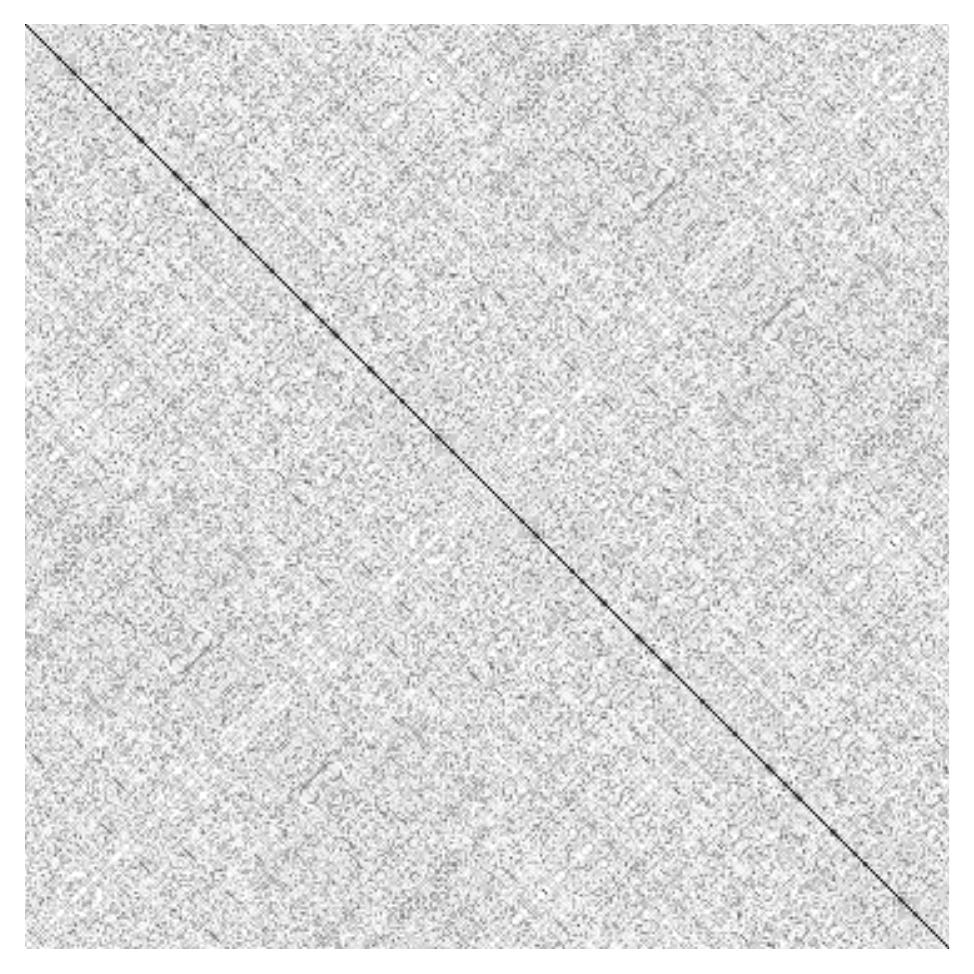}&
\includegraphics[width=0.21\columnwidth]{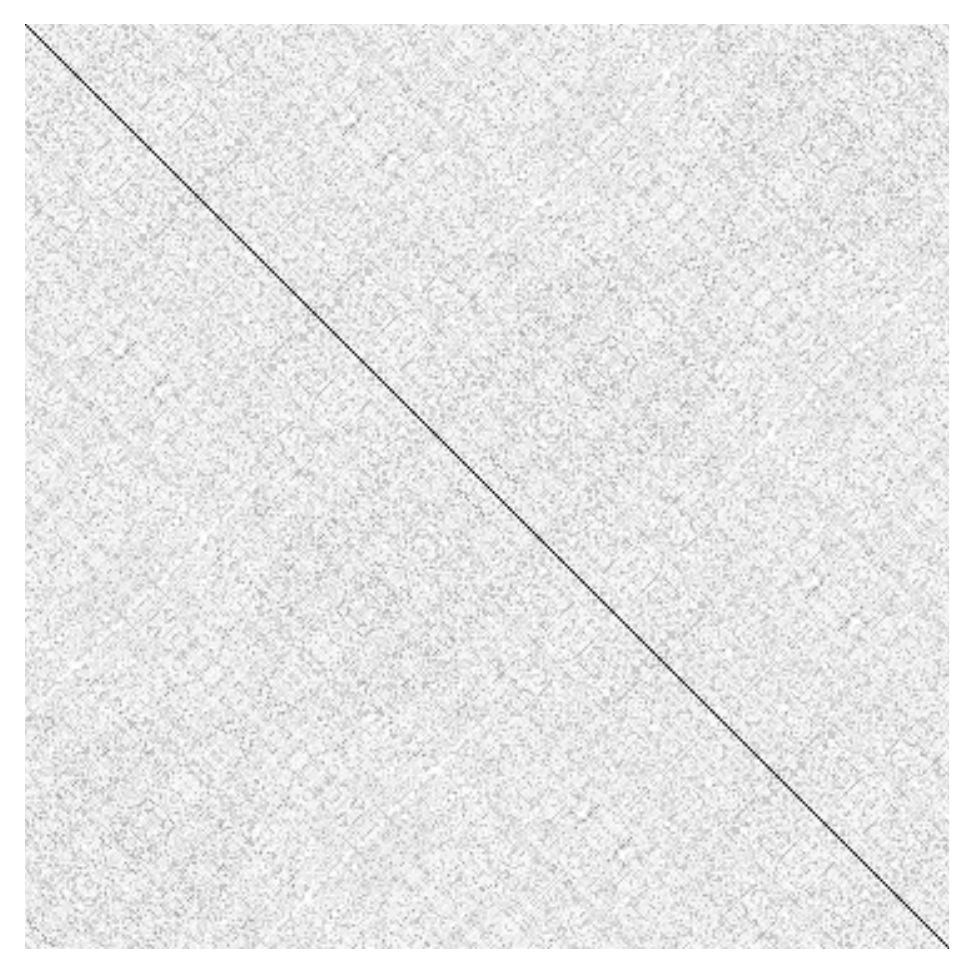}&
\includegraphics[width=0.21\columnwidth]{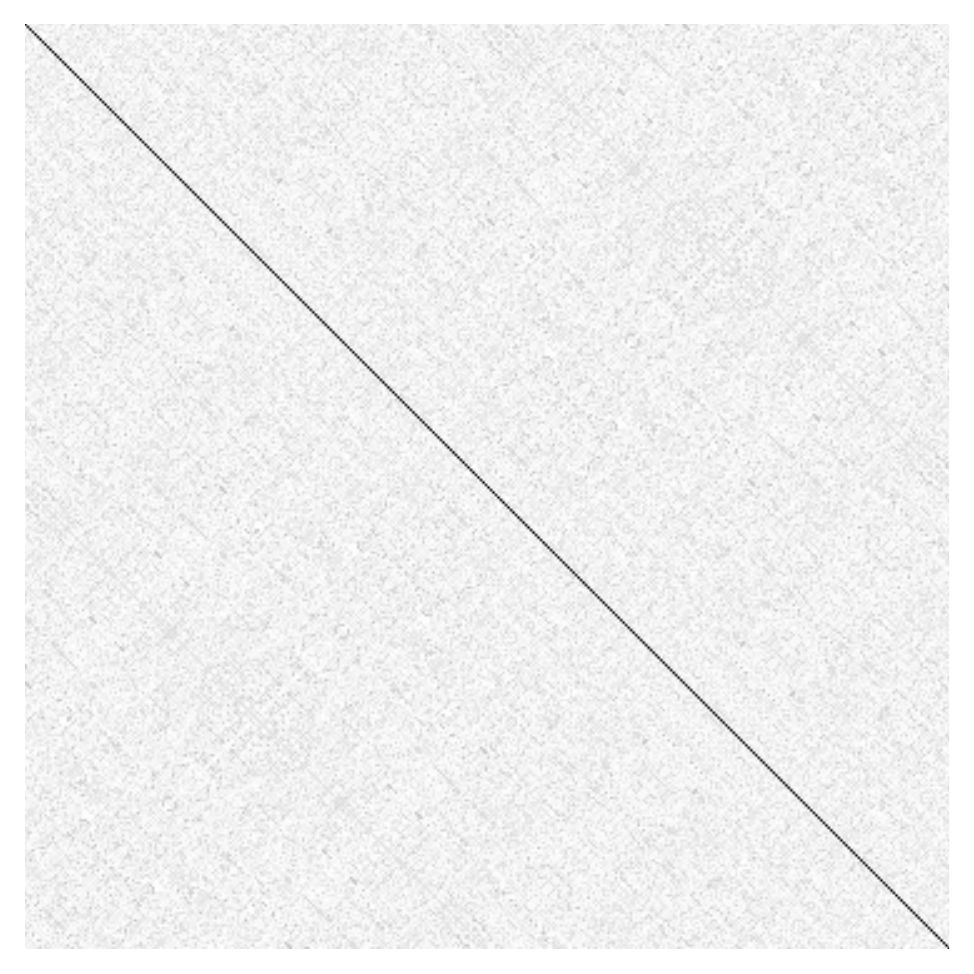}&
\includegraphics[width=0.21\columnwidth]{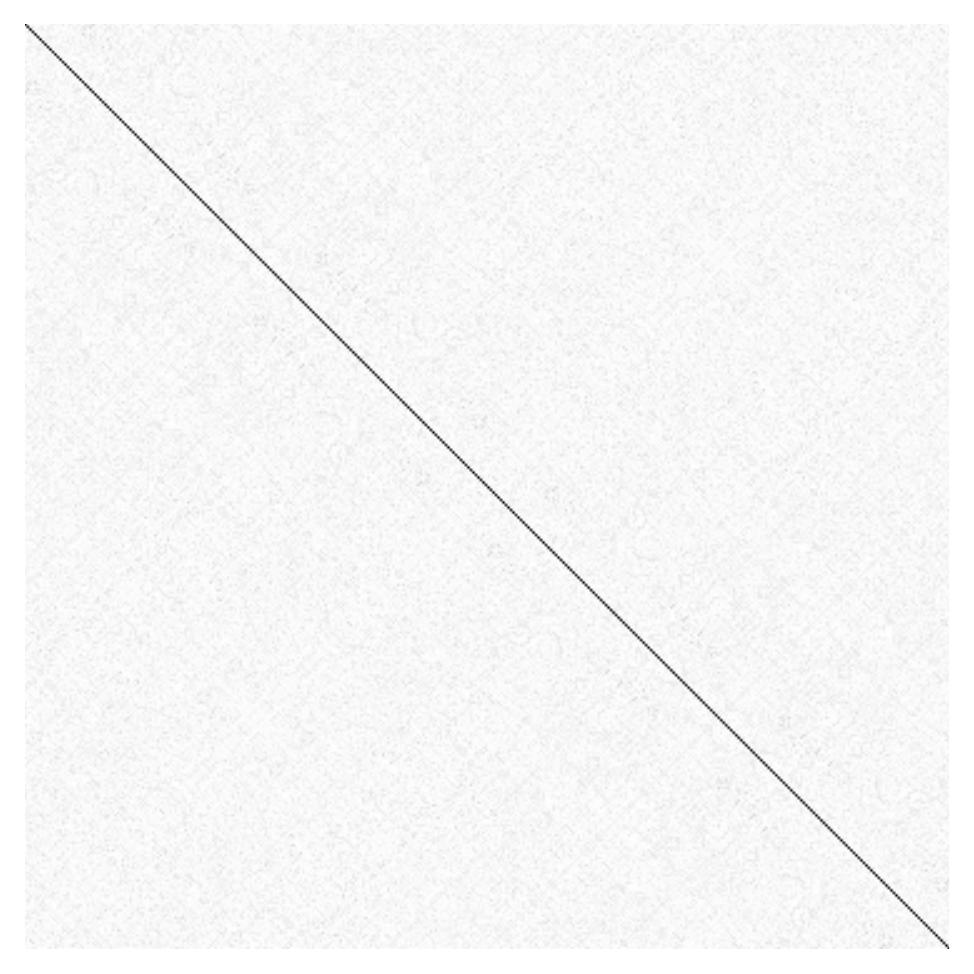}
\end{tabular}
\caption{Diagonalization masks for 12 qubits, obtained by averaging the
  outer products of commutation vectors $d_q$ using, from left to
  right, 30, 100, 1000, and 3000 random
  $q\in\mathcal{X}$.}\label{Fig:Diagonalization}
\end{figure}

For a given $q$, the term $D_q M D_q$ can be written as the
elementwise product of $M$ and the outer product $d_qd_q^T$. The outer
product has the important property that diagonal elements are always
one, irrespective of the signs in $d_q$. For randomly sampled
$q\in\mathcal{X}$ or $q\in \mathcal{P}$, each off-diagonal value has
an equal chance of being plus or minus one, and therefore has an
expected value of zero. When the number of qubits $n$ is small we can
iterate over all possible $q$ values and obtain exact diagonalization
of $M$. For larger values of $n$ this becomes intractable and we can
therefore only approximately diagonalize $M$, as illustrated in
Figure~\ref{Fig:Diagonalization}. For calibration we estimate
$e_i M \mathbf{1} = M_{i,i} + e_iM(\mathbf{1}-e_i)$. In order to
suppress the second term we need to sample sufficiently many circuits
instances. For the actual estimate of $M_{i,i}$ itself we simply need
to sample sufficiently many times regardless of the circuit instance.
At this point we should also remark that the actual procedure depends
on measurements from the probability vector
$\tilde{p}=A[H^{-1},0]D_q\kett{\rho_0}$ and we therefore need to
sample each circuit sufficiently many times.

In this section we give bounds on the number of circuit instances we
need to estimate $M_{i,i}$ to a given accuracy. Given that we multiply
by an approximately diagonal mask, this bound depends in part on the
maximum off-diagonal elements in $M$. We show how this corresponds to
properties of the transfer matrix $A$ and study these properties for
different types of transfer matrices.
%
Our final estimates are given by the ratio of two quantities and we
therefore consider how the estimation error in these quantities affect
the result. Note that through this section we work with the full
matrix representation for clarify; as shown in
Section~\ref{Sec:Approach}, processing itself is done based on
individual elements.

For the the number of circuit instances we need to approximately
diagonalize $M$, we have the following result:

\begin{theorem}\label{Thm:SampleComplexity}
  Given $k$ randomly sampled values $q_1,\ldots,q_k \in\mathcal{X}$
  and an index set $\mathcal{I} \subseteq [2^n]$. Define
\[
\hat{M} = \frac{1}{k}\sum_{\ell=1}^k (D_{q_{\ell}}M D_{q_{\ell}})
\quad\mbox{and}\quad
\beta = \max_{i\in\mathcal{I}}\sum_{ j\neq i}\vert M_{i,j}\vert.
\]
Then we satisfy
$\vert e_i \hat{M} (\mathbf{1}-e_i)\vert \leq \epsilon$ and
$\vert \hat{M}_{i,i} - M_{i,i}\vert \leq \epsilon$ simultaneously for
all $i\in\mathcal{I}$ with probability at least $1-\delta$ whenever
\begin{equation}\label{Eq:SampleComplexity}
k \geq 2\frac{\log(2/\delta) + n\log(2) + \log(\vert\mathcal{I}\vert)}{\epsilon^2/(1+\beta)^2}.
\end{equation}
\end{theorem}
\begin{proof}
Let $i$ be any element in $\mathcal{I}$.
For scaling of the off-diagonal elements we uniformly sample $X$ from
$\{-1,1\}$. If we can ensure that each element is scaled by a factor
at most $\epsilon_b$, then we have an additive term with magnitude at
most $\epsilon_b\beta$ in the estimation of $M_{i,i}$. For the
estimation of $M_{i,i}$ itself, we apply~\eqref{Eq:HoeffdingA} with
$X$ following an appropriate distribution on $[-1,1]$ and a maximum
deviation of $\epsilon_a$. Using a union bound over the off-diagonal
elements in the row we obtain the condition
\begin{equation}\label{Eq:ConditionDelta}
2\exp(-\half k\epsilon_a^2) + 2(2^n-1)\exp(-\half k\epsilon_b^2) \leq \delta
\end{equation}
In case $\beta = 0 $ we can choose $\epsilon_a = \epsilon$ and let
$\epsilon_b \to \infty$. Using a union bound over the rows in
$\mathcal{I}$ gives a sufficient number of circuit instances of
\[
k \geq \frac{\log(2/\delta) + \log(\vert\mathcal{I}\vert)}{\epsilon^2/2}.
\]
For the more general case where $\beta \neq 0$ we choose
$\epsilon_a = \epsilon_b$, which reduces condition~$\eqref{Eq:ConditionDelta}$
to
\begin{equation}\label{Eq:UnionBoundA}
2^{n}\exp(-\half k\epsilon_a^2) = \exp\big(n\log(2)-\half k\epsilon_a^2\big)
\leq \delta/2.
\end{equation}
In order to satisfy $\epsilon_a + \epsilon_b\beta \leq \epsilon$ we
must choose $\epsilon_a \leq \epsilon/(1+\beta)$. Combined with a
union bound, obtained by multiplying the left-hand side
of~\eqref{Eq:ConditionDelta} by the cardinality of $\mathcal{I}$, this
gives the sample complexity stated in~\ref{Eq:SampleComplexity}.
\end{proof}

As an aside, we note that diagonalization of quantum noise channels
using Pauli twirls follows exactly the same principle as the one we
use for diagonalizing $M$. A simple modification of
Theorem~\ref{Thm:SampleComplexity} can then be used to determine the
number of circuits needed to ensure that all off-diagonal noise terms
are bounded by $\epsilon$.

\subsection{Example transition matrices}\label{Sec:ExampleMatrices}

For a given transition matrix $A$ we define a corresponding
transformed matrix $M = HAH^{-1}$. This can be seen as readout
transition matrix for \PauliZ\ operators. It is easily seen that
$M^{-1} = HA^{-1}H^{-1}$ whenever the inverse of $A$ exists. For
convex combinations of two error channels, namely
$A = \mu A_1 + (1-\mu)A_2$ with $\mu\in[0,1]$, we have
$M = \mu M_1 + (1-\mu)M_2$. This straightforwardly generalizes to the
convex combination of any number of transition matrices.

As a simple example of a transition matrix, consider the case where
the outcome of each qubit is independently flipped with some
probability $r$. The transition matrix for a single qubit is then
given by
\begin{equation}\label{Eq:BitflipChannel}
A_i = \left(\begin{array}{cc}1-r_i & s_i \\ r_i & 1-s_i\end{array}
\right),
\end{equation}
with $r=s$. These matrices are then combined into a global transition
matrix $A = A_{s_1}\otimes \cdots\otimes A_{s_n}$. The corresponding
Pauli readout transition matrix then has a particularly simple
structure:
\begin{equation}\label{Eq:BitFlipM}
HAH^{-1}
= \bigotimes_{\ell=1}^{n} (H_2 A_{r_{\ell}} H_2^{-1})
 = \bigotimes_{\ell=1}^{n} \left[
\begin{array}{cc} 1 & 0 \\ 0 & (1-2r_{\ell})
\end{array}\right]
\end{equation}
In this case, since $M$ is already diagonal, we do not need to shrink
the off-diagonal elements. It therefore suffices to choose an
arbitrary but fixed value for $q\in\mathcal{X}$ for the circuits,
rather than sample it.  Choosing $q=0$ simplifies the resulting
circuits. Assume for simplicity that all probabilities $r_{\ell}$ are
equal to $r$, then it follows from~\eqref{Eq:BitFlipM} that the
diagonal element $M_{i,i}$ is directly related to the weight of the
\PauliZ\ operator $P_i$. For each $\sigma_z$ term in $P_i$ we have a
multiplicative term $(1-2r)$. The diagonal term for $P_i$ with $k$
non-identity term is then given by $(1-2r)^k$. The term $(1-2r)^k$ is
bounded below by $1-2kr$, that means that for 30 qubits with $1\%$
probability of a measurement flip, the diagonal elements in $M$ are
still at least $0.4$.  In the noiseless case the bit flip probability
is zero and we obtain $A=M=I$.

The transition matrix for the case where we only measure zeros is
given by $A = e_0e^T$ with a corresponding matrix $M = ee_0^T$. In
case each outcome is measured with equal probability, regardless of
the state, we have $A = 2^{-n} ee^T$ and $M = e_0e_0^T$. Although not
realistic by themselves, these matrices could be used in convex
combinations with other transition matrices. A good, but not very
realistic, example of a transition matrix that is perfectly invertible
but provides difficulty for our method is the following permutation
matrix:
\[
A = \left[\begin{array}{cccc}
\cdot & \cdot & 1 & \cdot\\
1 & \cdot & \cdot &\cdot\\
\cdot & \cdot & \cdot & 1\\
\cdot & 1 & \cdot & \cdot
\end{array}\right],\qquad
HAH^{-1} = \left[\begin{array}{cccc}
1 & \cdot & \cdot & \cdot\\
\cdot & \cdot & -1 &\cdot\\
\cdot & 1 & \cdot & \cdot\\
\cdot & \cdot & \cdot & -1
\end{array}\right]
\]
If we have access to an approximate inverse $\hat{A}^{-1}$, we can
adjust our scheme to work with $D_qH\hat{A}^{-1}\tilde{p}$ instead of
$D_qH\tilde{p}$. This form of preconditioning could help increase the
magnitude of the diagonal elements in $M$ or reduce that of the
off-diagonal elements but may be computationally expensive.

\subsection{Practical considerations}

In most of the discussion so far we have assumed ideal state
preparation. Suppose that, instead of $\rho_0 = \ket{0}\bra{0}$, we
can only prepare $\tilde{\rho}_0$. For calibration this means that,
after diagonalization of $M$, we obtain the vector
\begin{equation}\label{Eq:NoisyRho0}
(M \odot I)Z(\hat{\rho}_0) = \diag(Z(\hat{\rho}_0))m
\end{equation}
rather than $m$. If we assume that state preparation for qubits it
independent and that each qubit $\ell$ is initialized to state
$(1-\alpha_{\ell})\ket{0}\bra{0}+\alpha\ket{1}\bra{1} = \half(I +
(1-2\alpha_{\ell})\sigma_z)$, then we have
\[
Z(\hat{\rho}_0)^T = \bigotimes_{\ell}\left(\begin{array}{c}1\\
                            1-2\alpha_{\ell}\end{array}\right).
\]
Under this assumption, that means that, if we can estimate the
$\alpha_{\ell}$ values, we can incorporate this information
in~\eqref{Eq:NoisyRho0} to better estimate $m$.

Once the calibration data set has been acquired it can be used to
mitigate readout errors for circuits with various $U$, possibly with
basis changes. In practical systems we can expect gradual changes in
systemic gate and readout errors. That means that calibration data has
a limited lifetime. For error mitigation in the proposed approach we
traverse the calibration data whenever we want to compute the
correction factor for an individual \PauliZ\ operator. This approach
makes updates to the calibration data set very light weight: we could
simply augment the calibration data with time stamps and periodically
add some new data points while retiring data that falls outside the
current time window. For approaches based on explicit inversion of the
transfer matrix, any such update would amount to regeneration of the
entire matrix and its inverse. The computation complexity for updating
the correction factor using~\eqref{Eq:Processing} is linear in the
size of the data set. The evaluation of an element in the Hadamard
matrix and commutation between two $n$-qubit Pauli operators both take
$\mathcal{O}(n)$ time.

Note that the scalar terms $H_{m,i}$ in~\eqref{Eq:Processing} could be
replaced by elements from any other matrix, say $G$ with
$\vert G_{m,i}\vert \leq 1$, provided that the (relevant) diagonal
elements of $GAH^{-1}$ are sufficiently large.

As mentioned in Section~\ref{Sec:Analysis}, we can only access
information about $M$ by sampling from
$\tilde{p} = AH^{-1}D_qZ(\rho_0)$ for each instance $q$. In practice
we therefore need to make a tradeoff between the number of circuit
instances and the number of samples per circuit. We leave a detailed
analysis of this for future work.

\begin{figure*}[t!]
\centering
\begin{tabular}{ccccc}
\includegraphics[height=102pt]{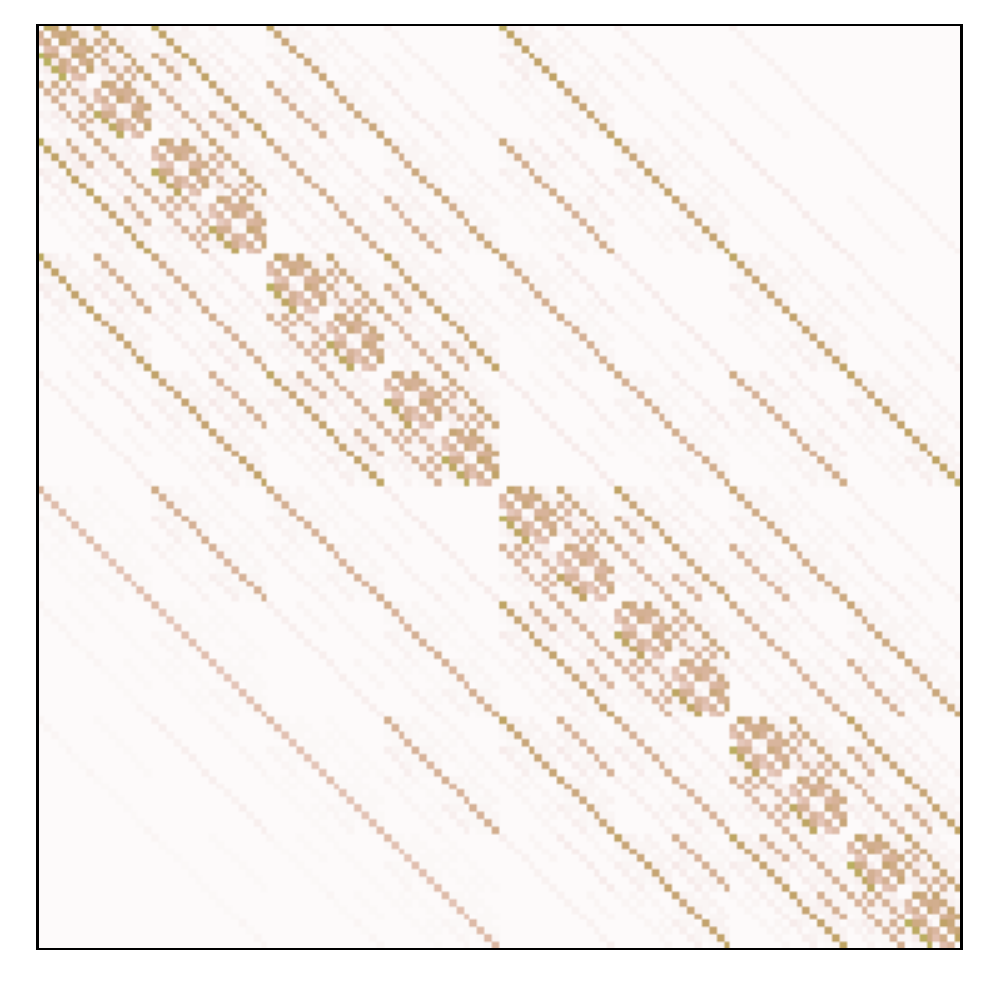}&
\includegraphics[height=102pt]{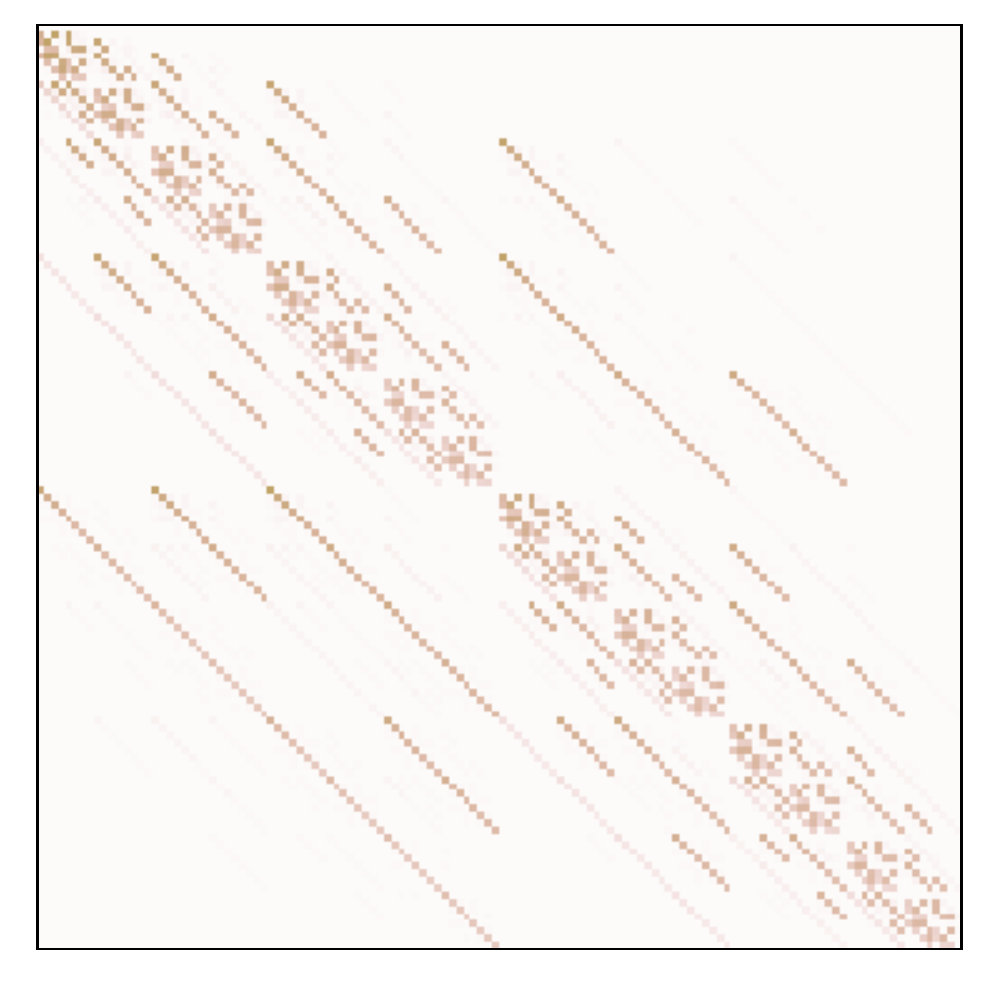}&
\includegraphics[height=102pt]{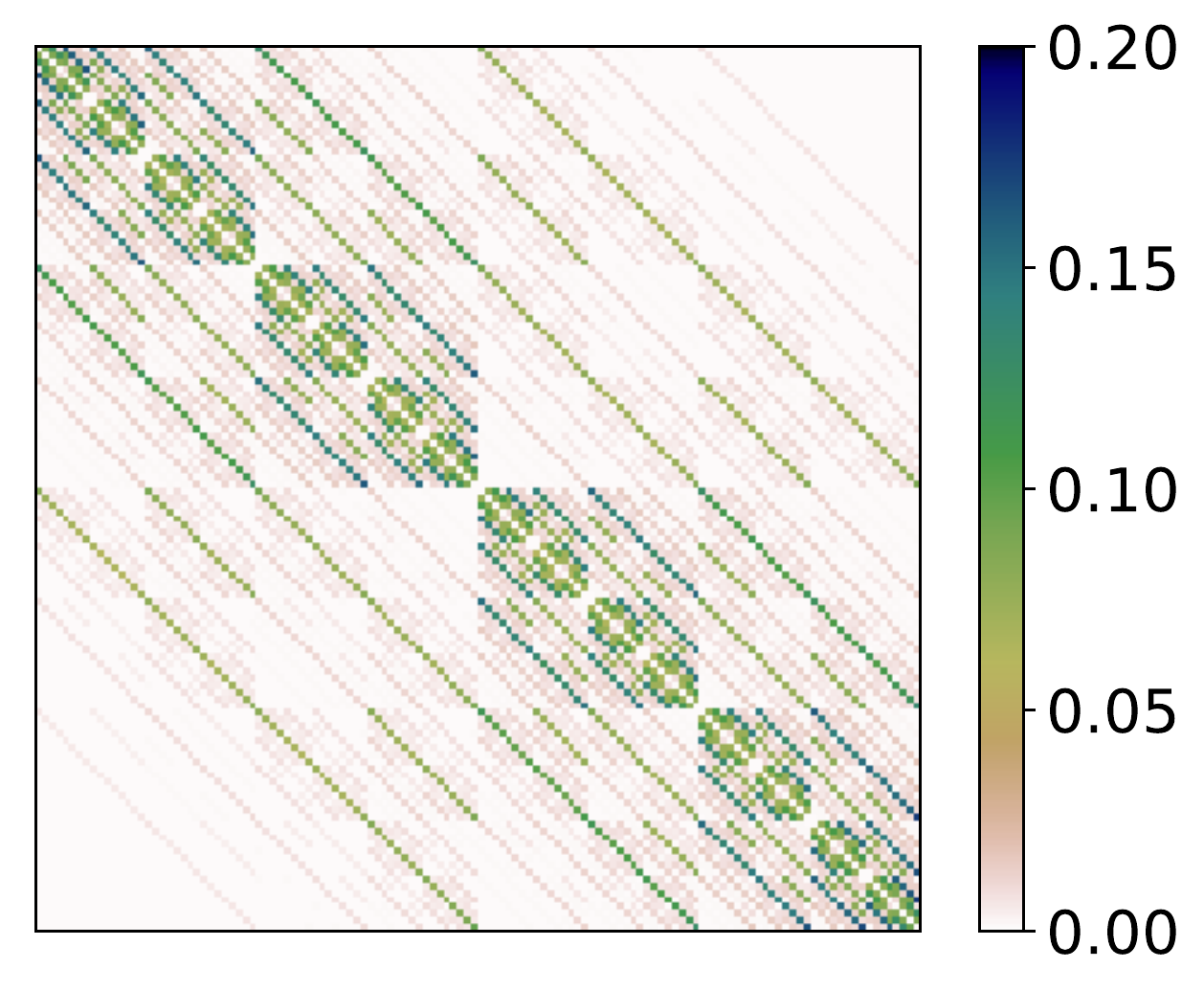}& \hspace*{8pt}&
\includegraphics[height=102pt]{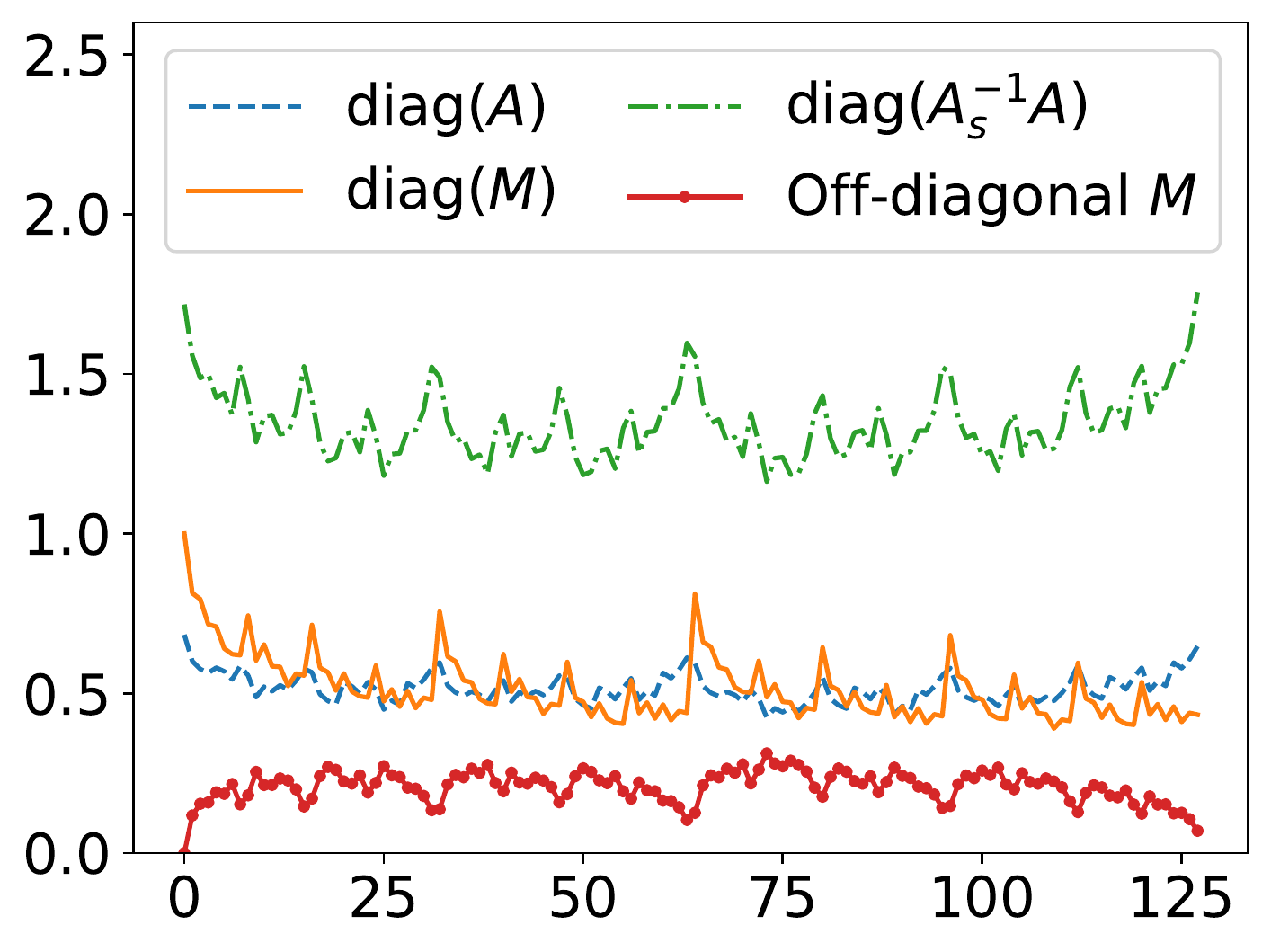}\\
({\bf{a}}) & ({\bf{b}}) & ({\bf{c}}) && ({\bf{d}})
\end{tabular}
\caption{Magnitude of the off-diagonal matrix entries of (a)
  transition matrix $A$, (b) matrix $M = HAH^{-1}$, and (c) matrix
  $A_s^{-1}A$, where $A_s$ is the single-qubit bit-flip model of $A$.
  We zero out the diagonal elements to highlight the off-diagonal
  structure and relative magnitude of the elements; the diagonal
  matrix elements are shown in plot (d) along with the sum of absolute
  values of the off-diagonal elements in $M$. If the transition matrix
  $A$ did not include any correlated readout errors, the matrix
  $A_s^{-1}A$ would be the identity matrix with unit diagonal
  entries.}\label{Fig:ErrorMatrices}
\end{figure*}

\begin{figure*}
\centering
\begin{tabular}{ccc}
\includegraphics[height=116pt]{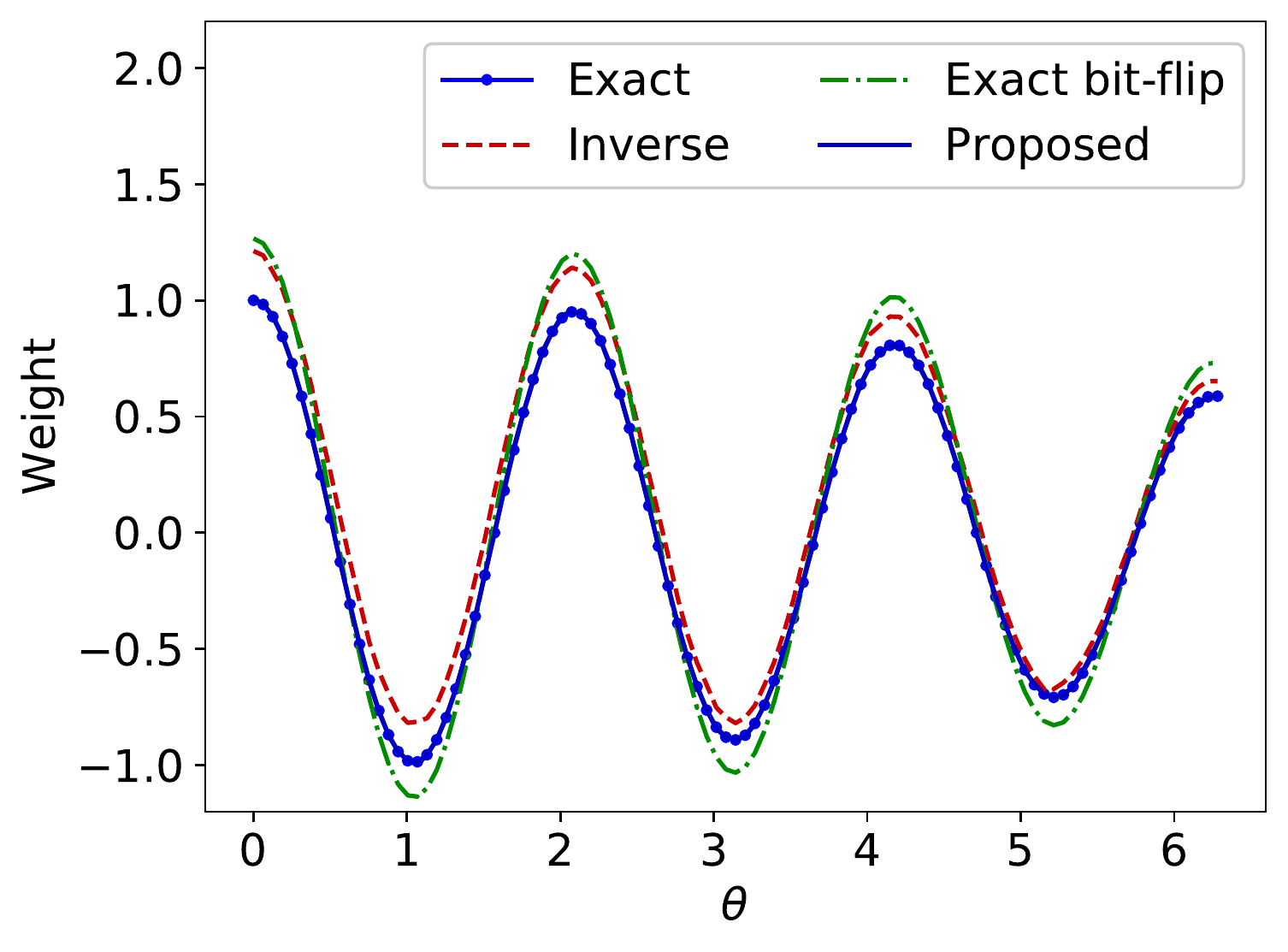}&
\includegraphics[height=116pt]{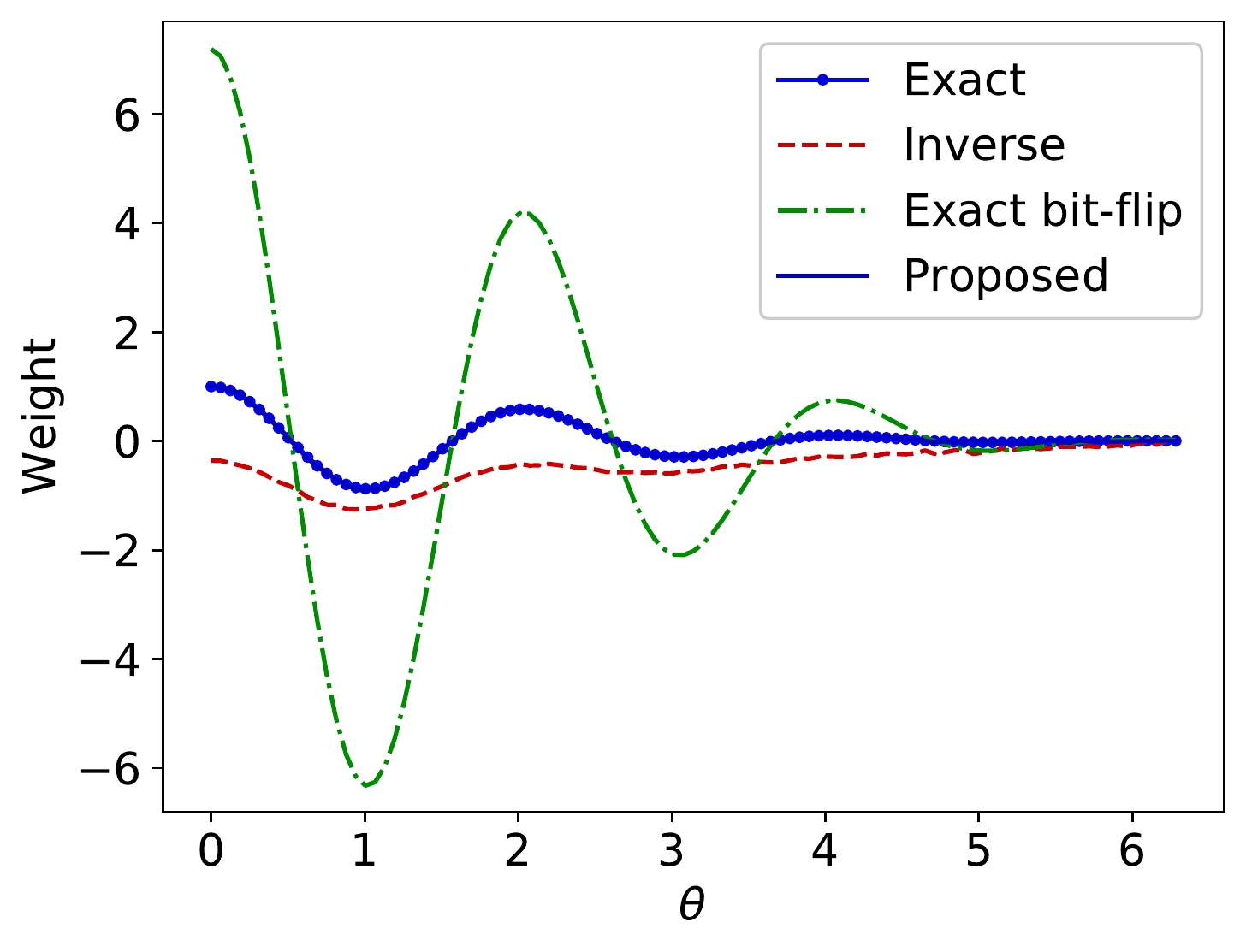}&
\includegraphics[height=116pt]{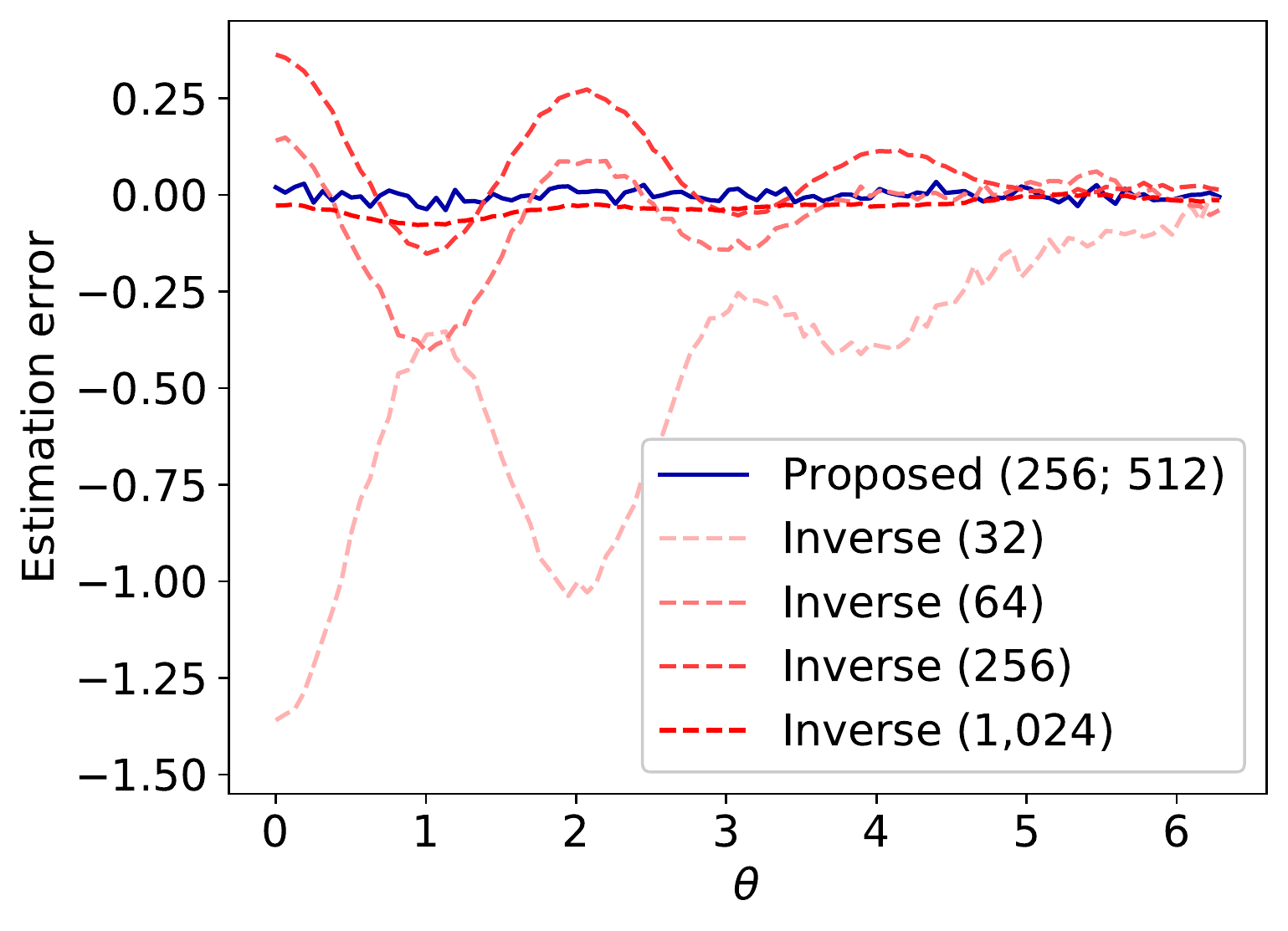}\\[-3.5pt]
({\bf{a}}) & ({\bf{b}}) & ({\bf{c}})\\[8pt]
\includegraphics[height=116pt]{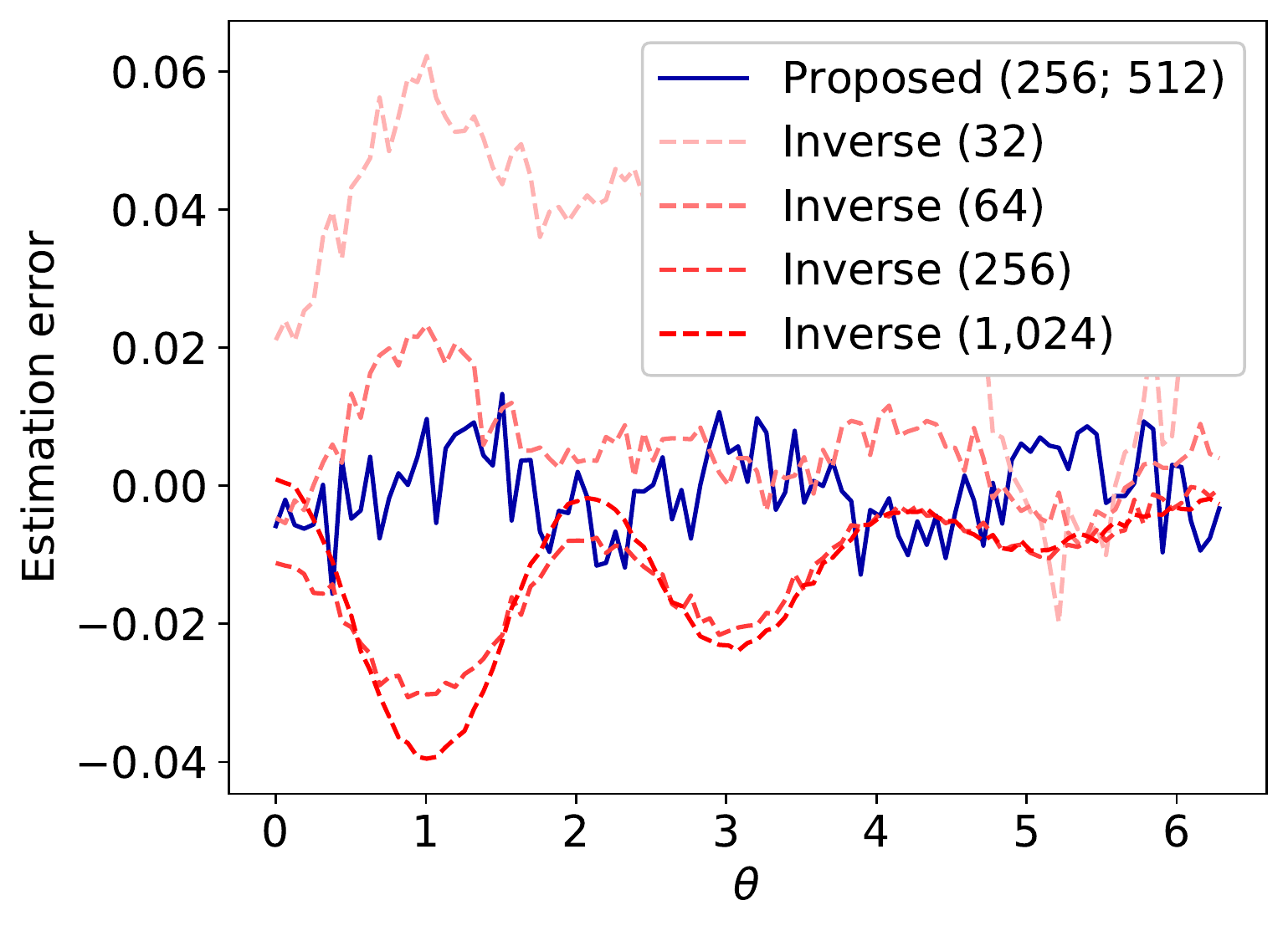}&
\includegraphics[height=116pt]{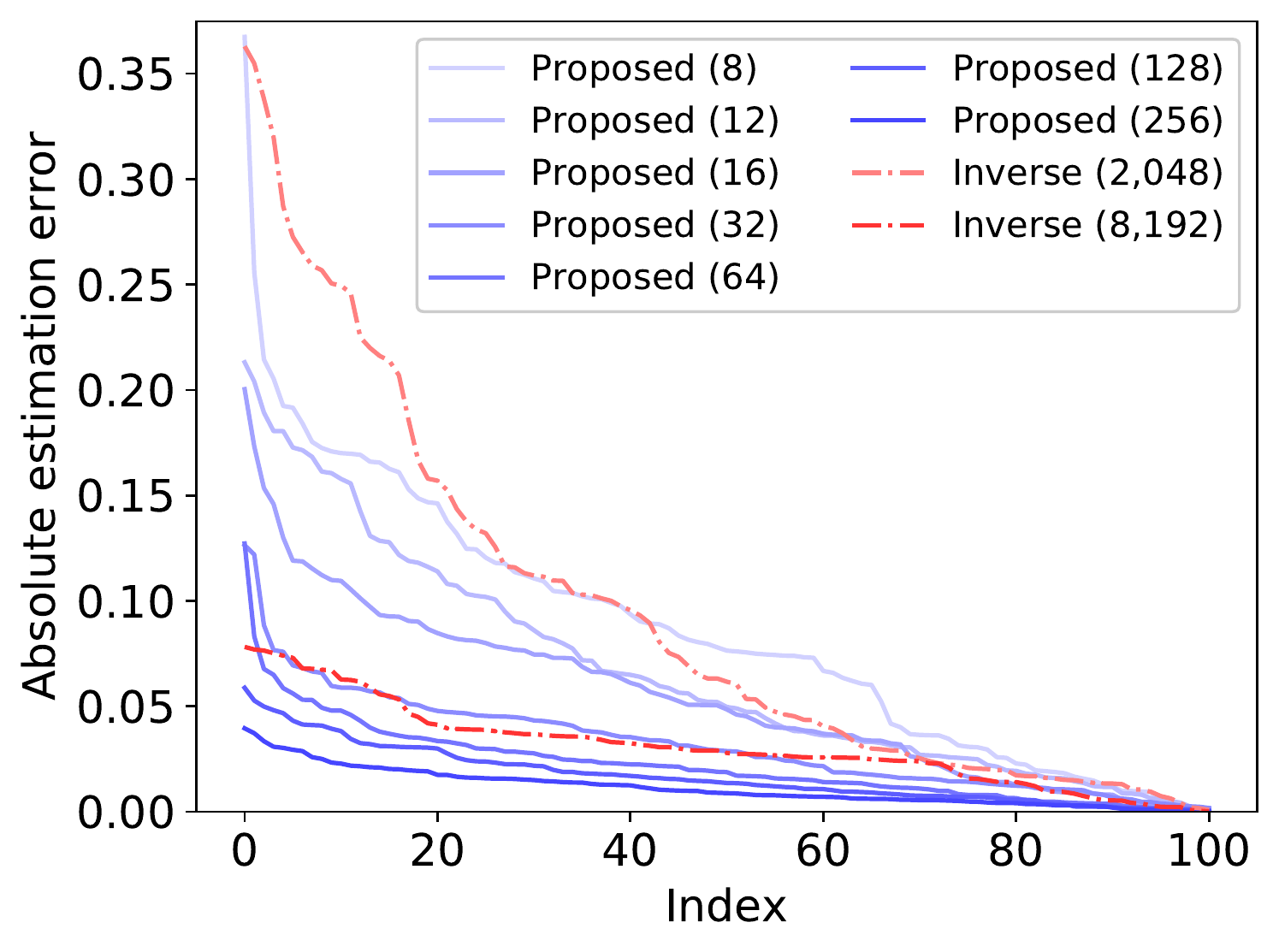}&
\includegraphics[height=116pt]{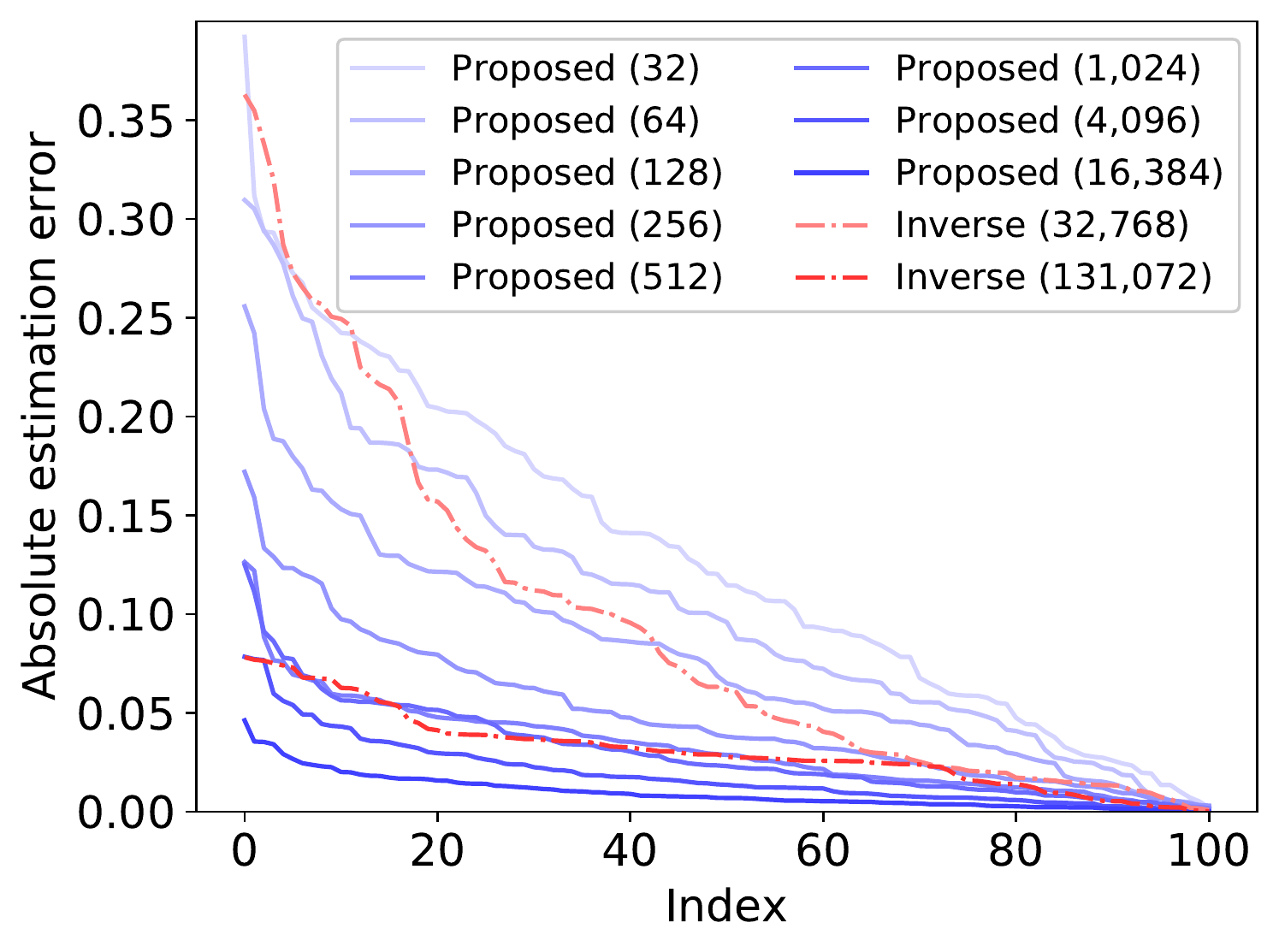}\\[-3.5pt]
({\bf{d}}) & ({\bf{e}}) & ({\bf{f}})
\end{tabular}
\caption{Simulation results for exact and estimate weights for
  \PauliZ\ operators in a twelve-qubit system with (a) a single
  $\sigma_z$ term on the first qubit, and (b) all $\sigma_z$ terms
  using 256 random circuits and 512 measurement per circuit for the
  proposed method and the same total number of measurements for the
  matrix-inversion approach. We evaluate the method for a range of
  $\theta$ values, resulting in different solutions as given
  by~\eqref{Eq:Weights}.  Plots (c) and (d) give the estimation error
  for the Pauli operators in (b) and (a), respectively, for easy
  comparison. The proposed method again uses 256 random circuits with
  512 measurements each. For the matrix-inversion approach we sample
  each of the $2^{12}$ circuits 32 to 1,024 times. Keeping the number
  of measurements fixed to 512, plot (e) shows the sorted
  approximation errors over the $\theta$ value, for different numbers
  of circuits. For the inverse approach we show the equivalent number
  of circuits such that the product gives the total number of
  measurements. Plot (f) keeps number of circuits in the proposed
  approach fixed to 32 and varies the number of measurements per
  circuit. For the inverse approach we list the equivalent number as
  before.}\label{Fig:Results}
\end{figure*}

\section{Simulation}\label{Sec:Simulation}

We evaluate the performance of the proposed method on a simple quantum
circuit consisting of a single $R_y(\alpha_i \theta)$ gate on each
qubit, where $\alpha_i$ is a scaling parameter that differs for each
qubit, and $\theta$ is a global phase. We choose $\alpha_1 = 3$ and
$\alpha_i = 0.15$ for all remaining qubits. The Pauli transfer matrix
for $R_y(\theta) = \mathrm{exp}(-i\theta Y/2)$ is given by
\[
T_{R_y(\theta)} =
\left(\begin{array}{cccc}
1 & 0 & 0 & 0 \\
0 & \cos(\theta) & 0 & -\sin(\theta)\\
0 & 0 & 1 & 0\\
0 & \sin(\theta) & 0 & \cos(\theta)
\end{array}\right).
\]
When applied to a \PauliZ\ operator with $\sigma_z$ components for
qubits $i\in\mathcal{I}$, the final weight is given by
\begin{equation}\label{Eq:Weights}
\prod_{i\in\mathcal{I}} \cos(\alpha_i\theta)
\end{equation}
We simulate noisy readout by forming transition matrix $A$ with
individual asymmetric bit-flip channels with a larger weight for $1$
to $0$ transitions. The transition matrix additionally includes
correlated readout errors on pairs of qubits. In
Figures~\ref{Fig:ErrorMatrices}(a)--(c) we illustrate a seven-qubit
transition matrix $A$, the corresponding $M = HAH^{-1}$, as well as
the product $A_s^{-1}A$, where $A_s$ is the transition matrix that we
would obtain if we would determine the exact bit-flip frequencies for
each of the qubits, and form the corresponding transition matrix. The
diagonal terms of the aforementioned three matrices, along with the
sum of absolute values of the off-diagonal elements in $M$, are
plotted in Figure~\ref{Fig:ErrorMatrices}(d).

As a first experiment we consider a 12-qubit system. For the proposed
method we sample 256 circuits each with 512 measurements.  For matrix
inversion we take the same total number of measurements, but spread
out over each of the 2,048 columns, each corresponding to a unique
circuits, thus giving a maximum of 64 measurements per circuit. The
resulting estimates for the weights of \PauliZ\ operators with
$\sigma_z$ at the first qubit, respectively all qubits are shown in
Figures~\ref{Fig:Results}(a) and~\ref{Fig:Results}(b) for a range of
$\theta$ values in~\eqref{Eq:Weights}. The estimates obtained using
the proposed approach are very close to the exact solution; so close
in fact that the curves are hard to distinguish. Given that the
transition matrix contains correlated noise terms, the exact bit-flip
approximation can never exactly mitigate the readout noise, as seen
from the rather poor performance. Finally, the results based on the
inverse of the estimated transition matrix $\hat{A}$ appear to be
biases in both settings and relatively lead more accurate, but still
nowhere near the performance of the proposed method. By increasing the
number of samples per circuit we can improve the accuracy of the
estimates, as illustrated in Figure~\ref{Fig:Results}(c)
and~\ref{Fig:Results}(d). However, even with 32 times more
measurements, the results obtained using matrix inversion are still
not as accurate as those obtained using the proposed method. Note that
the estimation error for the low-weight Pauli operator in
Figure~\ref{Fig:Results}(d) is much lower the the weight-$n$ Pauli.

In the next set of experiments we fix the number of circuits for
the proposed method to 512 and vary the number of measurements per
circuit instance. We sort the resulting estimation errors in magnitude
for the different $\theta$ values and plot the result in
Figure~\ref{Fig:Results}(e). We compare the results with those
obtained using the matrix-inversion approach, and display the
equivalent number of measurements per circuits. Forming the full
matrix requires $2^{12} = 2,048$ circuits, which means that the actual
number of measurements per circuit is four times lower than the number
shown. To match the performance of the proposed method with 32 samples
per circuit, the matrix-inversion approach requires an equivalent of
8,192 samples per circuit, which is 128 times more measurements in
total.  Similarly,
In Figure~\ref{Fig:Results}(f) we compare the
performance of the methods by fixing the number of circuits to 32 and
varying the number of measurements per circuit. 

\section{Conclusions}
In this work we have proposed an efficient, yet exceedingly simple,
method for readout error mitigation in the estimation of Pauli
observables. Unlike most existing techniques, the proposed approach
does not require any {\it{a priori}} assumptions or model of the
readout-error process. The approach is based on the augmentation of
quantum circuits with randomly selected Pauli operators and the
evaluation of a scalar function based on the measurements obtained
using a series of random instances. Readout errors are then mitigated
simply by dividing the function value for the quantum circuit of
interest by that of the benchmark circuit. The approach works by
diagonalizing the readout-error transfer matrix in the Hadamard
domain, which makes it trivial to invert. In contrast to many of the
existing algorithms, the proposed approach directly estimates the
weight of the \PauliZ{} components in the state, rather than the
distribution of the distribution of measurement values. Our
simulations show that the method is capable of mitigating correlated
readout error in a twelve-qubit system with very few measurements and
circuit instances.

\smallskip After completion of the manuscript we became aware of the
independent work by Chen {\it{et al.}}~\cite{CHE2020YZFa-arXiv} that
contains similar ideas as those presented here.


\begin{acknowledgments} 
  We thank Sergey Bravyi for helpful comments and discussion. This
  work is supported by the IBM Research Frontiers Institute.
\end{acknowledgments}


\bibliography{bibliography}
\end{document}